\documentclass[16pt,  onesided, charterfonts]{patmorin}
\usepackage{setspace}
\usepackage{latexsym}
\usepackage{amsmath,amsthm, amsfonts,graphicx,calc}
\usepackage{graphicx}
\usepackage{color}
\definecolor{grey}{gray}{0.4}

\renewcommand{\sc}{\textsc}

{}%
{}

\newlength{\la}
\newlength{\lx}
\newlength{\lb}
{\newenvironment{toocol}{\begin{longtable}{@{\hspace{\lx}}p{\la}@{\hspace{\lx}}p{\lb}}}{\end{longtable}}}
{}

{\newenvironment{edu}{\begin{longtable}{@{\hspace{\lx}}p{\la}@{\hspace{-1cm}}p{5.5in}}}{\end{longtable}}}
{}

{\newenvironment{serve}{\begin{longtable}{@{\hspace{\lx}}p{\la}@{\hspace{-3cm}}p{5.5in}}}{\end{longtable}}}
{}

{\newenvironment{titlecol}{\begin{longtable}{cc}}{\end{longtable}}}
{}

\newcounter{marker}
\newcounter{save}

\newenvironment{fortuna}{\begin{list}
        {$\bullet$}
        {\setlength{\itemsep}{-.17cm}
         \setlength{\topsep}{0cm}
         \setlength{\parskip}{0cm}
         \setlength{\labelwidth}{8mm}
	 \setlength{\labelsep}{1mm}
 	 \setlength{\leftmargin}{9mm}
         \usecounter{marker}
         \setcounter{marker}
        {\value{save}}}}
        {\setcounter{save}{\value{marker}}
\end{list}}

\newtheorem{theorem}{Theorem}
\newtheorem{lemma}[theorem]{Lemma}
\newtheorem{corollary}[theorem]{Corollary}

\newtheorem{open}[theorem]{Open Problem}

\newcommand{\figlabel}[1]{\label{fig:#1}}

\newcommand{\figref}[1]{\mbox{Figure~\ref{fig:#1}}}
\newcommand{\thmlabel}[1]{\label{thm:#1}}
\newcommand{\thmref}[1]{Theorem~\ref{thm:#1}}
\newcommand{\lemlabel}[1]{\label{lem:#1}}
\newcommand{\lemref}[1]{Lemma~\ref{lem:#1}}

\newcommand{\Figure}[4][htb]{
\begin{figure}[#1]
  \vspace*{1ex}
  \begin{center}#3\end{center}
        \vspace*{-2ex}
        \caption{\figlabel{#2}#4}
\end{figure}}

\def\R{{\mathbb{R}}}
\def\N{{\mathbb{N}}}

\def\G{{\mathcal{G}}}
\def\HH{{\mathcal{H}}}

\newcommand{\etal}{\emph{et al.}}
\newcommand{\ch}[1]{\ensuremath{\protect\textup{CH}(#1)}}

\newcommand{\xx}{\ensuremath{\protect{v_1}}}
\newcommand{\yy}{\ensuremath{\protect{v_2}}}
\newcommand{\zz}{\ensuremath{\protect{v_n}}}

\newcommand{\eg}{\ensuremath{G}}
\newcommand{\fg}{\ensuremath{\mathcal{F}}}

\newcommand{\po}{\ensuremath{<_\mathcal{F}}}

{\makeatletter
 \gdef\xxxmark{%
   \expandafter\ifx\csname @mpargs\endcsname\relax 
     \expandafter\ifx\csname @captype\endcsname\relax 
       \marginpar{xxx}
     \else
       xxx 
     \fi
   \else
     xxx 
   \fi}
 \gdef\xxx{\@ifnextchar[\xxx@lab\xxx@nolab}
 \long\gdef\xxx@lab[#1]#2{{\bf [\xxxmark #2 ---{\sc #1}]}}
 \long\gdef\xxx@nolab#1{{\bf [\xxxmark #1]}}
 \long\gdef\xxx@lab[#1]#2{}\long\gdef\xxx@nolab#1{}%
}

\title{A Center Transversal Theorem for Hyperplanes and
  Applications to Graph Drawing\footnote{The preliminar version of this paper has appeared in the Proceedings of the $27$th annual ACM symposium on Computational geometry, SoCG, 2011.}}

\author{Vida Dujmovi{\'c}\footnote{School of Computer Science, Carleton University, Ottawa, Canada  (\texttt{vida@scs.carleton.ca}). Partially supported by NSERC.}$\ $ and$\ $ Stefan Langerman\footnote{D\'epartement d'Informatique, Universit\'e Libre de Bruxelles, Brussels, Belgiu (\texttt{stefan.langerman@ulb.ac.be}). Ma\^itre de recherches du F.R.S.-FNRS.}}

\date{}

\begin{document}

\maketitle
\begin{abstract}
Motivated by an open problem from graph drawing, we study several
partitioning problems for line and hyperplane arrangements. We prove a
ham-sandwich cut theorem: given two sets of $n$ lines in $\R^2$, there
is a line $\ell$ such that in both line sets, for both halfplanes delimited
by $\ell$, there are $\sqrt{n}$ lines which pairwise intersect in that
halfplane, and this bound is tight;  
a centerpoint theorem: for any set of $n$ lines there is a point such
that for any halfplane containing that point there are $\sqrt{n/3}$ of
the lines which pairwise intersect in that halfplane. 
We generalize those results in higher dimension and obtain a center
transversal theorem, a same-type lemma, and a positive portion 
Erd\H{o}s-Szekeres theorem for hyperplane arrangements.  
This is done by formulating a generalization of the center transversal
theorem which applies to set functions that are much more general than measures.
Back to Graph Drawing (and in the plane), we completely solve the open
problem that motivated our search: there is no set of $n$ labelled lines that are universal for all $n$-vertex labelled planar graphs. As a contrast, the main result by Pach and Toth in [{\em J. of Graph Theory}, 2004], has, as an easy consequence, that every set of $n$ (unlabelled) lines is universal for all $n$-vertex (unlabelled) planar graphs. 

\end{abstract}





\section{Introduction}
Consider a mapping of the vertices of a graph to distinct points in
the plane and represent each edge by the closed line segment between
its endpoints. Such a graph representation is a \emph{(straight-line) drawing} if the only vertices that each edge intersects are its own endpoints. A \emph{crossing} in a drawing is a pair of edges that intersect at some point other than a common endpoint. A drawing is \emph{crossing-free} if it has no crossings.

One main focus in graph drawing is finding methods to produce
drawings or crossing-free drawings for a given graph with various
restrictions on the position of the vertices of the graph in the
plane. For instance, there is plethora of work where vertices are required to be placed on integer grid points or on parallel lines in $2$ or $3$--dimensions.

Given a set $R$ of $n$ regions in the plane and an $n$-vertex graph
$G$, consider a class of graph drawing problems where $G$ needs to be
drawn crossing-free by placing each vertex of $G$ in one region of
$R$. If such a drawing exists, then $R$ is said to \emph{support}
$G$. The problems studied in the literature distinguish between two
scenarios: in one, each vertex of the graph is prescribed its specific
region (that is, the vertices and the regions are labelled); in the
other, each vertex is free to be assigned to any of the $n$ regions
(that is, the vertices are unlabelled).

When regions are points in the plane, Rosenstiehl and Tarjan \cite{rt-rplbo-86}
asked if there exists a set of $n$ points that support all $n$-vertex
unlabelled planar graphs. This question is answered in the negative by De
Fraysseix \cite{fpp-sssfe-88, dFPP90}. On the contrary, every set of $n$ points in general
position supports all $n$-vertex unlabelled outerplanar graphs, as proved by
Gritzmann \etal \cite{Gritzmann91} and recapitulated in Lemma $14.7$ in the text by Agarwal and Pach~\cite{jp}. 
If the drawings are not restricted to be straight-line, then every set of labelled points supports every labelled planar graph, as shown by \cite{pw-epg-01}. However $\Omega(n)$ bends per edge may be necessary in any such crossing-free drawing.


When regions are labelled lines in the plane, Estrella-Balderrama~\etal\cite{1531118} showed that for every $n\geq 6$, there is no set of $n$ \emph{parallel} lines in the plane that support all labelled $n$-vertex planar graphs. The authors moreover characterized a (sub)class of $n$-vertex planar graphs that are supported by every set of $n$-parallel lines, for every labelling of the graphs in the class. That class is mainly comprised of several special families of trees.
Dujmovi\'c~\etal\cite{mi-GD10} showed that no set of $n$ lines that
\emph{all} intersect in one common point supports all $n$-vertex labelled planar graphs.  Moreover, they
show that for every $n$ large enough, there is a set of $n$ lines in general
position that does not support all labelled $n$-vertex planar graphs. They
leave as the main open problem the question of whether, for every $n$
large enough, there exists a \emph{universal} set of $n$ lines in the
plane, that is, one that supports all labelled $n$-vertex planar graphs.  In
Section~\ref{sec:graph-drawing}, as our main graph drawing result, we
answer that question in the negative. 
The main result by Pach and Toth~\cite{pt04} on monotone drawings, has, as an easy consequence, that in the unlabelled case, every set of $n$-lines supports every $n$-vertex unlabelled planar graph. As a side note, we give an alternative and direct proof of that fact. The result illustrates the sharp contrast with the labelled case.

While the positive result is proved using little of the
geometry in the arrangement, the non-existence of universal line sets
required extraction of some (bad) substructure from any line arrangement.
This prompted us to study several structural and partitioning problems
for line and hyperplane arrangements.

\paragraph{Hyperplane arrangements}
Partitioning problems are central to our understanding of discrete and
computational geometry, and while many works have focused on
partitioning point sets, probability distributions or measures, much
less is understood for sets of lines in $\R^2$ or hyperplanes in
$\R^d$. This is partially due to the fact that a line (or a hyperplane), being infinite,
can't be contained in any bounded region, or even in a halfplane
(except if the boundary of the halfplane is parallel to the given
line). Previous works (such as cuttings \cite{cf-dvrsi-90,m-cen-90} or equipartitions
\cite{ls02-opt-in-arr}) have focused on identifying, and bounding the number of
lines/hyperplanes intersecting a set of regions. Others \cite{csss-otass-89} on
partitioning the vertices of the arrangements rather than the lines themselves. Those results
have found numerous applications.
Our graph drawing problem motivates a different approach.

An arrangement $L$ of $n$ lines in $\R^2$ is composed of vertices
$V(L)$ (all pairwise intersections between lines of $L$), edges
connecting these vertices, and half-lines. If we omit the half-lines,
we are left with a finite graph which can be contained in a bounded
region of the plane, in particular, it is contained in the convex hull
$CH(V(L))$ of the vertices of the arrangement. 
Therefore, a natural way of evaluating the portion of an arrangement
contained in a given convex region $C$ is to find the largest subset
$L'$ of lines of $L$ such that the arrangement of $L'$ (without the
half-lines) is contained in $C$, or equivalently, such that all
pairwise intersections of lines in $L'$ lie in $C$.

It is not hard to show that, in any arrangement of $n$ lines, a
 line $\ell$ can be found such that for both closed halfplanes
bounded by $\ell$ there are at least $\sqrt{n}$ lines which pairwise intersect in
that halfplane. This provides the analogue of a bisecting line for point sets.
In Section~\ref{sec:lines}, we show that any two line arrangements can be
bisected simultaneously in this manner, thus proving a ham-sandwich
theorem for line arrangements. We also prove a centerpoint theorem:
for any arrangement of $n$ lines, there is a point $q$ such that for any
halfplane containing $q$, there are at least $\sqrt{n/3}$ lines of the
arrangement that pairwise intersect in that halfplane. 
In Section~\ref{sec:hyperplanes} we generalize these notions to higher
dimensions and prove a center transversal theorem: for any $k$ and
$d$, there is a growing function $Q$ such that for any $k$ sets
$A_1,\ldots,A_k$ of hyperplanes in $\R^d$, 
there is a $(k-1)$-flat $\pi$ such that for any
halfspace $h$ containing $\pi$ there is a subset $A'_i$ of $Q(|A_i|)$
hyperplanes from each set $A_i$ such that any $d$ hyperplanes of
$A'_i$ intersect in $h$. The bound $Q$ we find is related to Ramsey
numbers for hypergraphs.

Ham-sandwich theorems have a number of natural consequences. In
Section~\ref{sec:cent-transv} we show a same-type lemma for hyperplane
arrangements: informally, for any $k$ arrangements $A_1,\ldots,A_k$ of 
hyperplanes in general position (no $d+1$ share a point) and that are
large enough, we can find a large subset of hyperplanes $A'_i$ from
each set $A_i$ such that the convex hulls $CH(A'_i)$ of the vertices
in the arrangements $A'_i$ are well-separated, that is, no hyperplane
hits $d+1$ of them. In the plane, we also show a positive portion
Erd\H{o}s-Szekeres theorem: for any integers $k$ and $c$ there is an
integer $N$ such that any set of $N$ lines in general position
contains $k$ subsets $A_1,\ldots,A_k$ of $c$ lines each such that the
vertices of each arrangement $A_i$ can be separated from those of all the
others by a line.

All the results above would be relatively easy to prove if the set
function we were computing -- the maximum subset of hyperplanes that
have all $d$-wise intersections in a given region -- was a measure. 
Unfortunately it is not. However, in Section~\ref{sec:cent-transv}, we
identify basic properties much weaker than those of measures which, 
if satisfied by a set function, guarantee a central-transversal
theorem to be true. 

\section{Center transversal theorem}\label{sec:cent-transv}
The center transversal theorem is a 
generalization of both the ham-sandwich cut theorem, and the
centerpoint theorem discovered independently by
Dol'nikov \cite{Dolnikov-center-transversal}, and {\v Z}ivaljevi{\'c}
and Vre{\'c}ica \cite{zv-ehst-90}. 
The version of Dol'nikov is defined for a class of set functions that
is more general than measures.
Let $\HH$ be the set of all open halfspaces in $\R^d$ and let $\G$ be
a family of subsets of $\R^d$ closed under union operations and that
contains $\HH$. 
A \emph{charge} $\mu$ is a finite set function that is defined for all set
$X\in\G$, and that is \emph{monotone} ($\mu(X) \leq \mu(Y)$ whenever
$X\subseteq Y$) and \emph{subadditive} ($\mu(X\cup Y) \leq \mu(X)+\mu(Y)$).
A charge $\mu$ is \emph{concentrated} on a set $X$ if for every halfspace
$h\in\HH$ s.t. $h\cap X = \emptyset$, $\mu(h) = 0$.
Dol'nikov shows\footnote{Dol'nikov actually shows a slightly more
  general theorem that allows for non-concentrated charges. For the
  sake of simplicity we only discuss the simplified version even
  though our generalizations extend to the stronger original result.}:

\begin{theorem} [Center transversal theorem \cite{Dolnikov-center-transversal}]
For arbitrary $k$ charges $\mu_i$, $i=1,\ldots,k$, defined on $\G$ and
concentrated on bounded sets, there exists a $(k-1)$-flat $\pi$ such that 
$$\mu_i(h) \geq \frac{\mu_i(\R^d)}{d-k+2}, i=1,\ldots,k,$$
for every open halfspace $h\in\HH$ containing $\pi$.
\end{theorem}

A careful reading of the proof of this theorem reveals that its
statement can be generalized, and the assumptions on $\mu_i$
weakened. We first notice that the subadditive property is only used
in the proof for taking the union of a finite number of halfspaces from $\HH$.
Therefore, define $\mu$ to be \emph{$\HH$-subadditive} if 
$$\mu(\cup_{h\in H}) \leq \sum_{h\in H} \mu(h)$$
for any finite set $H\subset \HH$ of halfspaces. 

Next, notice that in order for the proof to go through, the set
function $\mu$ need not be real-valued.
Recall \cite{bourbaki--a13} that a \emph{totally ordered unital magma} $(M,\oplus,\leq, e)$ is a
totally ordered set $M$ endowed with a binary operator $\oplus$ such
that $M$ is closed under $\oplus$ operations, $\oplus$ has \emph{neutral}
element $e$ (i.e., $x\oplus e = x = e \oplus x$) and is \emph{monotone} 
(i.e., $a \oplus c \leq b \oplus c$ and $c \oplus a \leq c \oplus b$ 
whenever $a \leq b$). Further, for all $x\in M$ and $c\in \N$, define
the \emph{$c^{th}$ multiple of $x$} as 
$cx := \oplus^c x := \underbrace{x\oplus (x \oplus( \ldots \oplus x)\ldots)}_{\mbox{$c$ times}}$. 

Then, it suffices that
$\mu$ take values over $M$, and use $e$ as the 0 used in the
definition of a concentrated set function above. It is then the
addition operator $\oplus$ which is to be used in the definition of
the subadditive (or $\HH$-subadditive) property and in the proof of
the theorem. 
Thus, just by reading the proof of Dol'nikov under this new light we have:

\begin{theorem} 
Let $\mu_i$, $i=1,\ldots,k$ be $k$ set functions defined on $\G$ and
taking values in a totally ordered unital magma $(M,\oplus,\leq, e)$. 
 If the functions $\mu_i$ are monotone $\HH$-subadditive and
concentrated on bounded sets, there exists a $(k-1)$-flat $\pi$ such that 
$$(d-k+2)\mu_i(h) \geq \mu_i(\R^d), i=1,\ldots,k,$$
for every open halfspace $h\in\HH$ containing $\pi$.
\end{theorem}

%

\xxx[Stefan]{Not sure yet if we should rewrite the proof for
  completeness.}

\section{Center transversal theorem for arrangements}\label{sec:arrangements}

Let $A$ be an arrangement of $n$ hyperplanes in $\R^d$. We write $V(A)$ for
the set of all vertices (intersection points between any $d$ hyperplanes) of $A$ and
$CH(A) = CH(V(A))$ for the convex hull of those points. In the
arguments that follow, by abuse of language, we will write $A$ and
mean $V(A)$ or $CH(A)$. For example, we say that 
the arrangement $A$ is \emph{above hyperplane $h$} when all points in
$V(A)$ are above $h$. More generally for a region $Q$ in 
$\R^d$, we say that the arrangement $A$ does not intersect $Q$ if
$CH(A)$ does not intersect $Q$. We say that the $k$ arrangements
$A_1, A_2, \ldots, A_k$ are \emph{disjoint} if their convex hulls do not
intersect. They are \emph{separable} if they are disjoint and no hyperplane
intersects $d+1$ of them simultaneously.

Let $\HH$ be the set of all open halfspaces in $\R^d$ and let $\G$ be
a family of subsets of $\R^d$ closed under union operations and that contains $\HH$.
For any set $S\in\G$,  let $\mu_A(S)$ be the
maximum number of hyperplanes of $A$ that have all their vertices
inside of $S$, that is,
$$\mu_A(S) = \max_{A'\subseteq A, V(A')\subseteq S} |A'|.$$ 
In particular, $\mu_A(\R^d) = \mu_A(CH(A)) = n$ and $\mu_A(\emptyset) = d-1$.

\subsection{Lines in $\R^2$}\label{sec:lines}
We start with the planar case. Thus, $A$ is a set of lines in $\R^2$,
and $\HH$ is the set of all open halfplanes.
Recall the Erd\H{o}s-Szekeres theorem \cite{es-cpg-35}
\begin{theorem}[Erd\H{o}s-Szekeres]\label{thm:erdos-szekeres}
For all integers $r$, $s$, any sequence of $n>(r-1)(s-1)$ numbers
contains either a non-increasing subsequence of length $r$
or an increasing subsequence of length $s$.
\end{theorem}

We show:
\begin{lemma}\label{lem:union-bound}
For any two sets $S_1\in \HH$ and $S_2 \in \G$, 
$$\mu_A(S_1\cup S_2) \leq \mu_A(S_1)\mu_A(S_2)$$
\end{lemma}

\begin{proof}
Let $\ell$ be a line defining two open halfplanes $\ell^+$ and $\ell^-$
such that  $S_1 = \ell^-$ and let $S_2' = S_2\setminus\ell^-$.
Rotate and translate the plane so that $\ell$ is the (vertical) $y$
axis, and $\ell^+$ contains all points with positive $x$ coordinate.
Let $A'$ be a maximum cardinality subset of $A$ such that
$V(A')\subseteq S_1\cup S_2$. Let $l_1,\ldots, l_{|A'|}$ be the
lines in $A'$ ordered by increasing order of their slopes, and let 
$Y = (y_1,\ldots,y_{|A'|})$ be the $y$ coordinates of the
intersections of the lines $l_i$ with line $\ell$, in the same order.
For any set $A_1\subseteq A'$ such that the $y_i$ values of the lines in
$A_1$ form an increasing subsequence in  $Y$,
notice that $V(A_1)\subseteq S_1$. Likewise, for any set
$A_2\subseteq A'$ that forms a non-increasing subsequence in $Y$, we  
have $V(A_2)\subseteq S_2'$. Any such set $A_1$ is of size
$|A_1|\leq\mu_A(S_1)$ and any such set $A_2$ is of size
$|A_2|\leq\mu_A(S_2')\leq \mu_A(S_2)$. 

Therefore, $Y$ has no non-decreasing subsequence of length
$\mu_A(S_1)+1$ and no non-increasing subsequence of length
$\mu_A(S_2)+1$, and so by Theorem~\ref{thm:erdos-szekeres}, 
$\mu_A(S_1\cup S_2) = |A'| = |Y|\leq\mu_A(S_1)\mu_A(S_2)$.
\end{proof}

\begin{corollary}
The set function $\mu_A$ takes values in the totally ordered unital magma
$(\R_0^+,\cdot,\leq, 1)$; it is monotone and $\HH$-subadditive. 
\end{corollary}

We can thus apply the generalized center transversal theorem with
$k=2$ to obtain a ham-sandwich cut theorem: 
\begin{theorem}\thmlabel{sep}
For any arrangements $A_1$ and $A_2$ of lines in $\R^2$, there exists a
line $\ell$ bounding closed halfplanes $\ell^+$ and $\ell^-$ and sets
$A_i^\sigma$, $i\in{1,2}$, $\sigma\in{+,-}$ such that 
$A_i^\sigma\subseteq A_i, |A_i^\sigma|\geq |A_i|^{1/2}$, and
$V(A_i^\sigma)\in \ell^\sigma$.    
\end{theorem}

Note that this statement is similar to the result of Aronov~\etal\cite{mut} on mutually avoiding sets. Specifically, two sets $A$ and $B$ of points in the plane are {\em mutually avoiding} if no line through a pair of points in $A$ intersects the convex hull of $B$, and vice versa. Note that, on the other hand, our notion of separability for lines is equivalent to the following definition in the dual. Two sets $A$ and $B$ of points in the plane are {\em separable} if there exists a point $x$ such that all the lines through pairs of points in $A$ are above $x$ and all the lines through pairs of points in $B$ are below $x$ or vice versa. Aronov~\etal\ show in Theorem $1$ of \cite{mut} that any two sets $A_1$ and $A_2$ of points contains two subsets $A_i'\subseteq A_i$, $|A_i'|\geq |A_i/12|^{1/2}$, $i\in\{1,2\}$ that are mutually avoiding. That this bound is tight, up to a constant, was proved by Valtr~\cite{val}. In the dual, \thmref{sep} states that for any two sets $A_1$ and $A_2$ of points in $\R^2$, there exists a point $\ell$ and sets
$A_i^\sigma$, $i\in{1,2}$, $\sigma\in{+,-}$ such that 
$A_i^\sigma\subseteq A_i, |A_i^\sigma|\geq |A_i|^{1/2}$, and  all lines through pairs of points in $A_i^+$ are above $\ell$ and all lines through pairs of points in $A_i^-$ are below $\ell$. While similar, neither the two results nor the two notions of mutually avoiding and separable are equivalent. It is not difficult to show that no result/notion immediately implies the other. Moreover, neither our proof of \thmref{sep} nor the proof of Theorem $1$ in \cite{mut}  give two sets that are, at the same time, mutually avoiding and separable. 

Note that the bound in \thmref{sep} is tight: assume $n$ is the
square of an integer. Construct the first line arrangement $A_1$ with
$\sqrt{n}$ pencils of $\sqrt{n}$ lines each, centered at points with
coordinates $(-1/2,i)$ for $i=1,\ldots,\sqrt{n}$, and the slopes of
the lines in pencil $i$ are distinct values in 
$[1/2-(i-1)/\sqrt{n},1/2-i/\sqrt{n}]$. 
Thus all intersections other than the pencil centers have $x$
coordinates greater than $1/2$. The line $x=0$ delimits two halfplanes
in which $\mu_{A_1}(x\leq 0) = \sqrt{n}$ since any set of more than
$\sqrt{n}$ lines have lines from different pencils which intersect on
the right of $x=0$, and $\mu_{A_1}(x\geq 0) = \sqrt{n}$ since any set
of more than $\sqrt{n}$ lines has two lines in the same pencil which
intersect left of $x=0$. Since $\mu_{A_1}$ is monotone, no vertical
line can improve this bound on both sides. Perturb the lines so that
no two intersection points have the same $x$ coordinate. For $A_2$, build a copy of
$A_1$ translated down, far enough so that no line through two vertices of
$A_1$ intersects $CH(A_2)$ and conversely. Therefore any line not
combinatorially equivalent to a vertical line (with respect to the vertices of $A_1$ and
$A_2$) does not intersect one of the arrangements and so there is no
better cut than $x=0$.

Applying the generalized center transversal theorem with
$k=1$ gives a centerpoint theorem with a bound of $|A|^{1/3}$. A
slightly more careful analysis improves that bound. 
\begin{theorem}
For any arrangement $A$ of lines in $\R^2$, there exists a point $q$
such that for every halfplane $h$ containing $q$ there is a set
$A'\subseteq A, |A'|\geq \sqrt{|A|/3}$, such that $V(A')\in h$.   
\end{theorem}

\begin{proof}
Let $H$ be the set of halfplanes $h$ such that $\mu_A(h)<  z = \sqrt{|A|/3}$.
The \emph{halfspace depth} $\delta(q)$ is the minimum value of
$\mu_A(h)$ for any halfspace containing $q$. Therefore, the region of
depth $\geq z$ is the intersection of the complements $\overline{h}$ of the
halfplanes $h\in H$.  
If there is no point of depth $\geq z$ then the intersection
of the complements of halfplanes in $H$ is empty, and so (by Helly's
Theorem) there must be 3 halfplanes $h_1$, $h_2$, and $h_3$ in $H$ 
such that the intersection of their complements $\overline{h_1}\cap\overline{h_2}\cap\overline{h_3}$ is empty. 
But then, there is at least one point $q\in h_1\cap h_2\cap h_3$. 
Let $h'_i$ be the translated halfplanes $h_i$ with point $q$ on the
boundary. Since $h'_i\subseteq h_i$, $\mu_A(h'_i)\leq\mu_A(h_i)<z$. 
The point $q$ and the 3 halfplanes through it are witness that there
is no point of depth $\geq z$. 

The 3 lines bounding those 3 halfplanes divide the plane into 6
regions. 
Every line misses one of the three regions $h'_1\cap \overline{h'_2}\cap\overline{h'_3}$, 
$\overline{h'_1}\cap h'_2\cap \overline{h'_3}$, and $\overline{h'_1}\cap \overline{h'_2}\cap h'_3$.
Classify the
lines in $A$ depending on the first region it misses, clockwise. The
largest class $A'$ contains $\geq |A|/3$ lines.
Assume without loss of generality that all lines in $A'$ miss 
$h'_1\cap \overline{h'_2}\cap \overline{h'_3}$, then all intersections between
lines of $A'$ are in $h'_2\cup h'_3$. By Lemma~\ref{lem:union-bound},
$$|A|/3\leq |A'|= \mu_{A'}(h'_2\cup h'_3)\leq
\mu_{A'}(h'_2)\mu_{A'}(h'_3)< z^2 = |A|/3,$$
a contradiction.
\end{proof}

\xxx[Stefan]{Show that it is tight.}
\xxx[Stefan]{Add separated $\alpha \beta$ cut?}

\subsection{Hyperplanes in $\R^d$}\label{sec:hyperplanes}
We first briefly review a bichromatic version of Ramsey's theorem for
hypergraphs.

\begin{theorem}
For all $p, a, b\in\N$, there is a natural number $R = R_p(a,b)$ such that for
any set $S$ of size $R$ and any 2-colouring $c:{S \choose p} \rightarrow \{1,2\}$ of all subsets
of $S$ of size $p$, there is either a set $A$ of size $a$ such that
all $p$-tuples in $A \choose p$ have colour 1 or a set $B$ of size $b$
such that all $p$-tuples in $B\choose p$ have colour 2.
\end{theorem}

\begin{lemma}\label{lem:Rd-union-bound}
For any two sets $S_1\in \HH$ and $S_2 \in \G$, 
$$\mu_A(S_1\cup S_2) \leq R_d(\mu_A(S_1)+1,\mu_A(S_2)+1)-1$$
\end{lemma}

\begin{proof}
Let $h$ be a hyperplane defining two open halfplanes $h^+$ and $h^-$
such that  $S_1 = h^-$ and let $S_2' = S_2\setminus h^-$.
Let $A'$ be a maximum cardinality subset of $A$ such that
$V(A)\subseteq S_1\cup S_2$. 
Colour every subset of $d$ hyperplanes in $A'$ with colour $1$ if their
intersection point is in $h^-$ and with colour $2$ otherwise. 

For any set $A_1\subseteq A'$ such that all subsets in 
$A_1 \choose d$ have colour $1$, notice that 
$V(A_1)\subseteq S_1$. Likewise, for any set
$A_2\subseteq A'$ such that all subsets in 
$A_2 \choose d$ have colour $2$, we  
have $V(A_2)\subseteq S_2'$. Any such set $A_1$ is of size
$|A_1|\leq\mu_A(S_1)$ and any such set $A_2$ is of size
$|A_2|\leq\mu_A(S_2')\leq \mu_A(S_2)$. 

Therefore, $A'$ has no subset of size $\mu_A(S_1)+1$ that has all
$d$-tuples of colour 1, and no subset of size
$\mu_A(S_2)+1$ that has all $d$-tuples of colour 2, and so by 
Ramsey's Theorem, 
$\mu_A(S_1\cup S_2) = |A'| \leq R_d(\mu_A(S_1)+1,\mu_A(S_2)+1)-1$.
\end{proof}

Define the operator $\oplus$ as $a \oplus b = R_d(a+1,b+1)-1$. The
operator is increasing and closed on the set $\N_{\geq d-1}$ of naturals $\geq d-1$. 
Since $R_d(d,x)=x$ for all $x$, $d-1$ is a neutral element. Therefore 
$(\N_{\geq d-1},\oplus,\leq,d-1)$ is a  totally ordered unital magma.
Thus we have:
\begin{corollary}
The set function $\mu_A$ takes values in the totally ordered unital magma
$(\N_{\geq d-1},\oplus,\leq,d-1)$; it is monotone and $\HH$-subadditive. 
\end{corollary}

Apply now the generalized center transversal theorem to obtain:
\begin{theorem}
Let $A_1, \ldots, A_k$ be $k$ sets of hyperplanes in $\R^d$. There
exists a $(k-1)$-flat $\pi$ such that for every open halfspace $h$ that
contains $\pi$, 
$$ (d-k+2)\mu_{A_i}(h) \geq |A_i|. $$  
\end{theorem}

The special case when $k=d$ gives a ham-sandwich cut theorem.

\begin{corollary}
Let $A_1, \ldots, A_d$ be $d$ sets of hyperplanes in $\R^d$. 
There exists a hyperplane $\pi$ bounding the two closed halfspaces
$\pi^+$ and $\pi^-$ and sets $A_i^\sigma\subseteq A_i$,  $\sigma \in \{+,-\}$, such that
$V(A_i^\sigma)\subseteq \pi^\sigma$ and  $|A_i^\sigma| \oplus |A_i^\sigma|\geq |A_i|$.
\end{corollary}

If the arrangement $A$ has the property that no $r+1$ hyperplanes intersect
in a common point, $\mu_A(\pi) \leq r$ for any hyperplane $\pi$, and
so by Lemma~\ref{lem:Rd-union-bound}, if $h$ is an open halfspace
bounded by $\pi$ and $\bar{h} = \pi\cup h$ is the corresponding closed
halfspace, $\mu_A(\bar{h}) \leq \mu_A(h) \oplus r$.

\begin{corollary}\label{cor:hamsw-genpos}
Let $A_1, \ldots, A_d$ be $d$ sets of hyperplanes in $\R^d$, no $r+1$ of
which intersect in a common point. 
There exists a hyperplane $\pi$ bounding the two open halfspaces
$\pi^+$ and $\pi^-$ and sets $A_i^\sigma\subseteq A_i$,  $\sigma \in \{+,-\}$, such that
$V(A_i^\sigma)\subseteq \pi^\sigma$ and  
$(|A_i^\sigma|\oplus |A_i^\sigma|) \oplus r \geq |A_i|$.
\end{corollary}

\section{Same-type lemma for arrangements} \label{sec:same-type-lemma}
Center transversal theorems, and especially the ham-sandwich cut
theorem, are basic tools for proving many facts in discrete
geometry. We show here how the same facts can be shown for hyperplane
arrangements in $\R^d$. 

A transversal of a collection of sets $X_1,\ldots,X_m$ is a $m$-tuple
$(x_1,\ldots,x_m)$ where $x_i\in X_i$. 
A collection of sets $X_1,\ldots,X_m$ has \emph{same-type
  transversals} if all of its transversals have the same
order-type.

Note that $m\geq d+1$ sets have same-type transversals
if and only if every $d+1$ of them have same-type transversals. 
There are several equivalent definitions for these notions.
\begin{enumerate}
\item The sets  $X_1,\ldots,X_{d+1}$ have same-type transversals if and only if
  they are \emph{well separated}, that is, if and only if for all disjoint sets of
  indices $I,J\subseteq\{1,\ldots,d+1\}$, there is a hyperplane
  separating the sets $X_i, i\in I$ from the sets $X_j, j\in J$.
\item Connected sets $C_1,\ldots,C_{d+1}$ have same-type transversals if and only if
  there is no hyperplane intersecting simultaneously all $C_i$. Sets
  $X_1,\ldots,X_{d+1}$ have same-type transversals if and only if there is no
  hyperplane intersecting simultaneously all their convex hulls $C_i=CH(X_i)$.
\end{enumerate}

The same-type lemma for point sets states that there is a constant
$c=c(m,d)$ such that for any collection
$S_1,\ldots,S_m$ of finite point sets in $\R^d$, there are sets 
$S'_i \subseteq S_i$ such that $|S'_i|\geq c|S_i|$ and the sets
$S'_1,\ldots,S'_m$ have same-type transversals.  We here show a
similar result for hyperplane arrangements.

A function $f$ is \emph{growing} if for any value $y_0$ there is a $x_0$
such that $f(x)\geq y_0$ for any $x\geq x_0$.

\begin{lemma}\label{lem:same-type-lemma}
For any integers $d$, $m$, and $r$, there is a growing function $f=f_{m,d,r}$ such that
for any collection of $m$ hyperplane arrangements $A_1,\ldots,A_m$,
in $\R^d$, where no $r+1$ hyperplanes
intersect at a common point, there are sets $A'_i\subseteq A_i$ such
that $|A'_i|\geq f(|A_i|)$ and the sets $CH(A'_1),\ldots,CH(A'_m)$
have same-type transversals.
\end{lemma}

\begin{proof}
The proof will follow closely the structure of Matou{\v s}ek
\cite[Theorem~9.3.1, p.217]{Matousek-2002-LDG}.
First notice that the composition of two growing functions is a
growing function. The proof will show how to choose successive
(nested) subsets of each set $A_i$, $c$ times where $c=c(m,d)$ only
depends on $m$ and $d$ and where the size of each subset is some
growing function of the previous one. 

Also, it will suffice to prove the theorem for $m = d+1$, and then
apply it repeatedly for each $d+1$ tuple of sets. The resulting
function $f_{m,d,r}$ will be $f_{d+1,d,r}^{m \choose {d+1}}$, the repeated
composition of $f_{d+1,d,r}$, ${m \choose {d+1}}$ times.

So, given $d+1$ sets $A_1,\ldots,A_{d+1}$ of hyperplanes in $\R^d$, 
suppose that there is an index set $I\subseteq\{1,\ldots,d+1\}$ such
that $\cup_{i\in I}\ch{A_i}$ and $\cup_{i\notin I}\ch{A_i}$ are not
separable by a hyperplane and assume without loss of generality that
$d+1\in I$.
Let $\pi$ be the ham sandwich cut hyperplane for
arrangements $A_1,\ldots,A_d$ obtained by applying
Corollary~\ref{cor:hamsw-genpos}. 
Then for each $i\in[1,d]$, each of the two open halfspaces $\pi^\sigma$,
$\sigma\in\{+,-\}$ bounded by $\pi$ contains a subset 
$A^\sigma_i\subseteq A_i$ such that $V(A_i^\sigma)\subseteq \pi^\sigma$ and  
$(|A_i^\sigma|\oplus |A_i^\sigma|) \oplus r \geq |A_i|$.
Furthermore, because $\mu_{A_{d+1}}(\pi)\leq r$ and by Lemma~\ref{lem:Rd-union-bound}, 
$$\mu_{A_{d+1}}(\pi^+)\oplus\mu_{A_{d+1}}(\pi^-)\oplus r
\geq\mu_{A_{d+1}}(\R^d) = |A_{d+1}|.$$

Assume without loss of generality 
$\mu_{A_{d+1}}(\pi^+)\geq \mu_{A_{d+1}}(\pi^-)$. Then 
$\mu_{A_{d+1}}(\pi^+)\oplus\mu_{A_{d+1}}(\pi^+)\oplus r \geq |A_{d+1}|$.
For each $i\in I$, let $A'_i = A^+_i$ and for each $i\notin I$, let 
$A'_i=A^-_i$. Let $g(x)=\min\{y|y\oplus y\oplus r \geq x\}$. Then $g$
is a growing function, and $|A'_i|\geq g(|A_i|)$.

In the  worst case, we have to shrink the sets for each possible $I$, $2^d$
times. Therefore for
$m=d+1$, the function $f$ in the statement of the theorem is a
composition of $g$, $2^d$ times, and is a growing function.
\end{proof}

In the plane, the same-type lemma readily gives a positive portion
Erd\H{o}s-Szekeres Theorem. Recall that the Erd\H{o}s-Szekeres (happy ending) theorem \cite{es-cpg-35}
states that for any $k$ there is a number $ES(k)$ such that any set of
$ES(k)$ points in general position in $\R^2$ contains a subset of size
$k$ which is in convex position.

\begin{theorem}\thmlabel{arrang}
For every integers $k$, $r$, and $c$, there is an integer $N$ such that
any arrangement $A$ of $N$ lines, such that no $r+1$ lines go through a
common point,
contains disjoint subsets $A_1,\ldots,A_k$ with $|A_i|\geq c$ and such
that every transversal of $\ch{A_1},\ldots,\ch{A_k}$ is in convex position.
\end{theorem}
\begin{proof}
Let $m=ES(k)$ and let $f=f_{m,2,r}$ be as in Lemma~\ref{lem:same-type-lemma}. 
Let $N$ be such that $f(\lfloor N/m\rfloor)\geq c$.
Partition the set $A$ of $N$ lines into $m$ sets $A_1,\ldots,A_m$ of
$N/m$ lines arbitrarily. 
Apply Lemma~\ref{lem:same-type-lemma} to obtain sets
$A'_1,\ldots,A'_m$ each of size at least $c$. Finally, choose one
transversal $(x_1,\ldots,x_m)$ from the sets $\ch{A'_i}$ and apply the
Erd\H{o}s-Szekeres theorem to obtain a subset $x_{i_1},\ldots,x_{i_k}$
of points in convex position. Because the sets $\ch{A'_i}$ have the
same type property, every transversal of
$\ch{A'_{i_1}},\ldots,\ch{A'_{i_k}}$ is in convex position.
\end{proof}

\pagebreak
\section{Graph Drawing}\label{sec:graph-drawing}

Formally, a vertex labelling of a graph $G=(V,E)$ is a bijection
$\pi:V\rightarrow [n]$. A set of $n$ lines in the plane labelled from
$1$ to $n$ \emph{supports $G$ with vertex labelling $\pi$} if there
exists a straight-line crossing-free drawing of $G$ where for each
$i\in[n]$, the vertex labelled $i$ in $G$ is mapped to a point on line $i$. A set $L$ of $n$ lines labelled from $1$ to $n$ \emph{supports} an $n$-vertex graph $G$ if for every vertex labelling $\pi$ of $G$, $L$ supports $G$ with vertex labelling $\pi$. 
In this context clearly it only makes sense to talk about planar
graphs. We are interested in the existence of an $n$-vertex line set that
supports all $n$-vertex planar graphs, that is, in the existence of
a \emph{universal} set of lines for planar graphs.

\begin{theorem}\thmlabel{gd-main}
For some absolute constant $c'$ and every $n\geq c'$, there exists no set of $n$ lines in the plane that support all $n$-vertex planar graphs.
\end{theorem}






The following known result will be used in the proof of this theorem.

\begin{lemma}\cite{mi-GD10}\lemlabel{core}
Consider the planar triangulation on $6$ vertices, denoted by $G_6$,
that is depicted on the bottom of \figref{G}. $G_6$ has vertex labelling $\pi$ such that the following holds for every set
$L$ of $6$ lines labelled from $1$ to $6$, no two of which are parallel. For
every straight-line crossing-free drawing, $D$, of $G_6$ where for each
$i\in[n]$, the vertex labelled $i$ in $\pi$ is mapped to a point on line
$i$ in $L$, there is a point that is in an interior face of $D$
and in \ch{L}. 
\end{lemma}

\begin{proof}[Proof of \thmref{gd-main}] 
Let $L$ be any set of $n\geq c'=5N$ lines, where $N$ is obtained from \thmref{arrang} with values $k=6$, $c=6$, $r=17$.

\cite{1531118} proved that for every $n\geq 6$, no set of $n$ parallel lines supports all $n$-vertex planar graphs. Thus if $L$ has at least $6$ lines that are pairwise parallel, then $L$ cannot support all $n$-vertex planar graphs.

\cite{mi-GD10} proved that for every $n\geq 18$, no set of $n$ lines that all go through a common point supports all $n$-vertex planar graphs. Thus if $L$ has at least $18$ such lines, then $L$ cannot support all $n$-vertex planar graphs.

Thus assume that $L$ has no $6$ pairwise parallel lines and no $18$ lines that intersect in one common point. Then $L$ has a subset $L'$ of $c'/5\geq N$ lines no two of which are parallel and no 18 of which go through one common point. Then \thmref{arrang} implies that we can find in $L'$ six sets $A_1, \dots, A_6$  of six lines each, such that the set $\{\ch{A_1}, \dots, \ch{A_6}\}$ is in convex position. Assume $\ch{A_1}, \dots, \ch{A_6}$  appear in
that order around their common ``convex hull''.

Consider an $n$-vertex graph $H$ whose subgraph $G$ is illustrated in \figref{G}. $G\setminus v$ has three
components, $A$, $B$, and $C$, each of which is a triangulation. Each of
the components $A$, $B$, and $C$ has two vertex disjoint copies of $G_6$
(the $6$-vertex triangulation from \lemref{core}). Map the vertices
of the first copy of $G_6$ in $A$ to $A_1$ and the second copy to
$A_4$ using the mapping equivalent to $\pi$ in \lemref{core}. 
Map the vertices of the first copy of $G_6$ in $B$ to $A_2$  and the
second copy to $A_5$ using the mapping  equivalent to $\pi$ in \lemref{core}. 
Map the vertices of the first copy of $G_6$ in $C$ to $A_3$  and the
second copy to $A_6$ using the mapping equivalent to $\pi$ in \lemref{core}. 
Map the remaining vertices of $H$ arbitrarily to the remaining lines
of $L$.

\begin{figure}[htb]
 \centerline{\begin{tabular}{c}
  \includegraphics[scale=0.45]{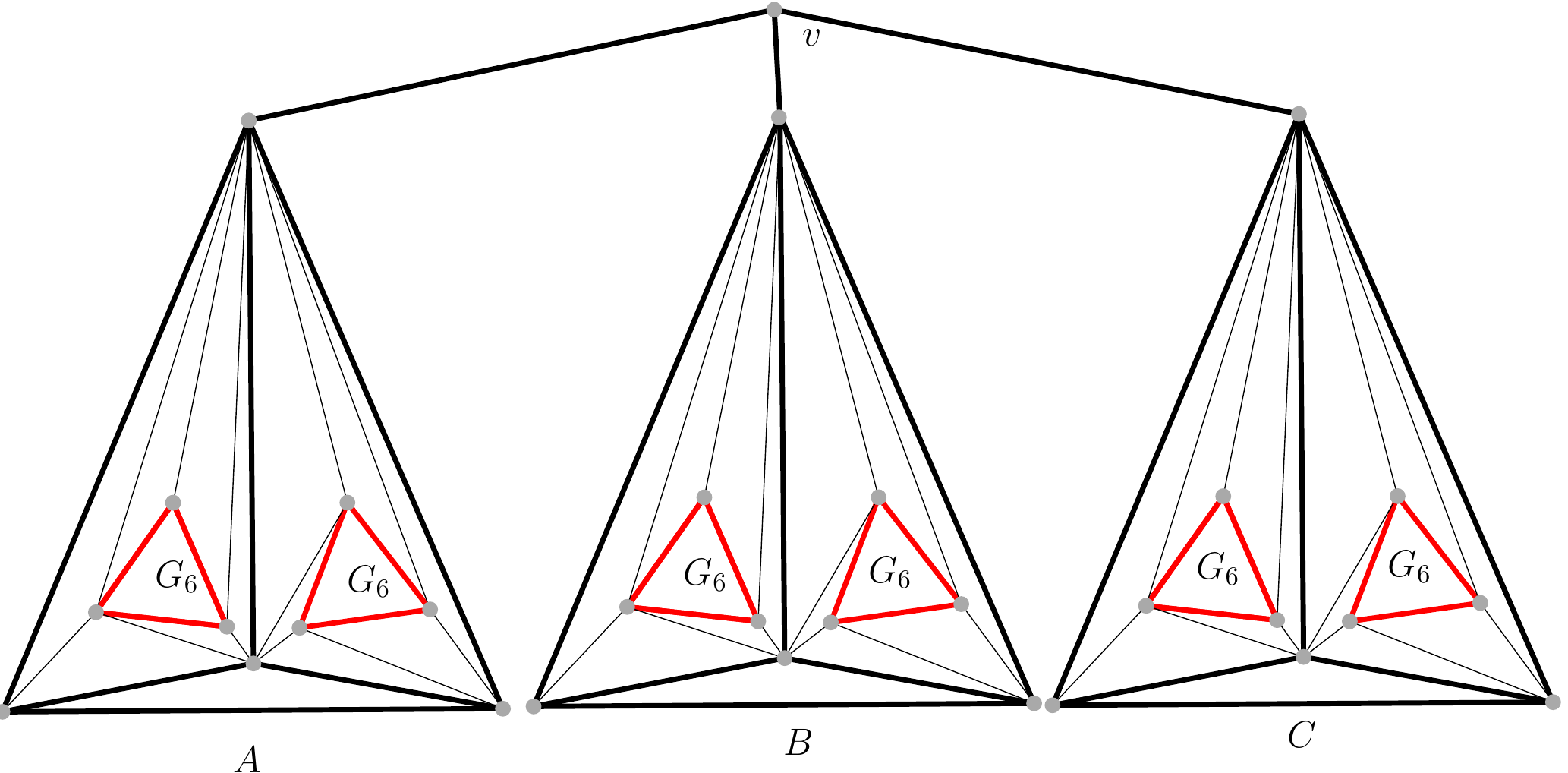}\\
\includegraphics[scale=0.45]{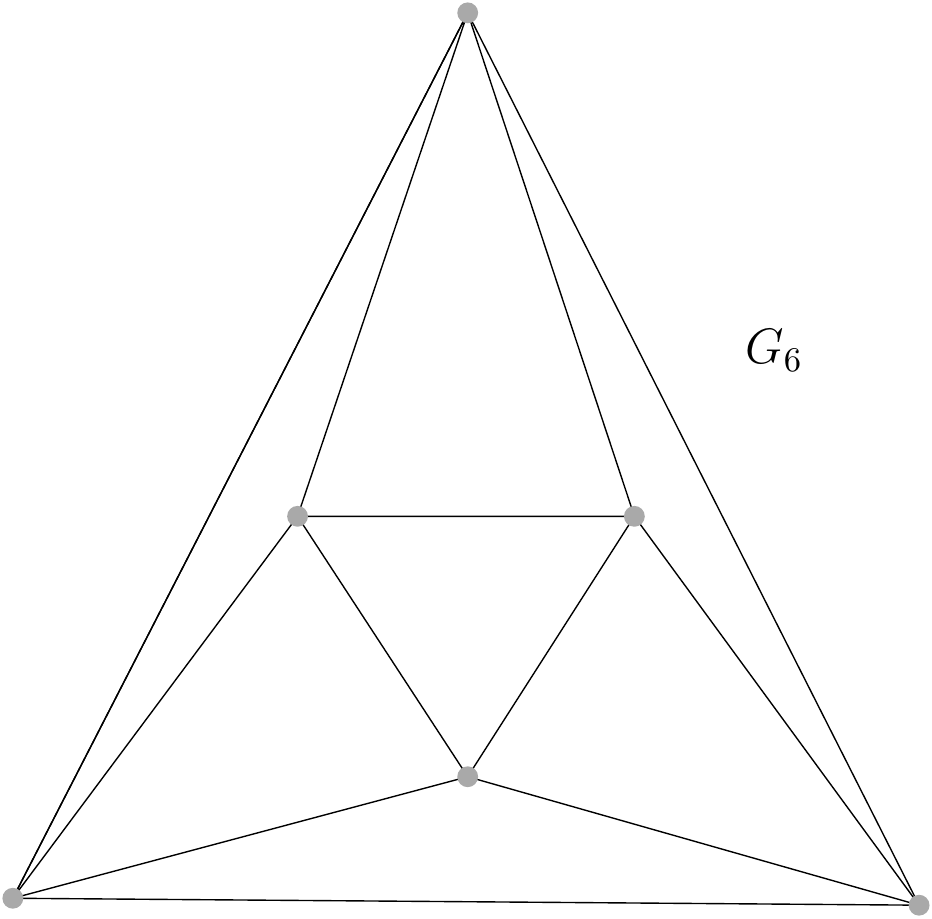}
\end{tabular}}
  \caption{\figlabel{G} Illustration for the proof of \thmref{gd-main}.}
\end{figure}

We now prove that $L$ does not support $H$ with such a
mapping. Assume, for the sake of contradiction, that 
 it does and consider the resulting crossing-free drawing $D$ of $H$.
In $D$ the drawing of each of $A$, $B$, and $C$ has a triangle as an
outerface. Let $T_A$, $T_B$, and $T_C$ denote these three triangles
together with their interiors in the plane.

It is simple to verify that in any crossing-free drawing of $G$ at
least two of these triangles are disjoint, meaning that there is no point
$p$ in the plane such that $p$ is in both of these triangles. 
Assume, without loss of generality, that $T_A$ and $T_B$ are disjoint. 
By \lemref{core}, there is a point $p_1\in \ch{A_1}$ such that
$p_1\in T_A$, and a point $p_4\in \ch{A_4}$ such that $p_4\in
T_A$. Thus the segment $\overline{p_1p_4}$ is in $T_A$. Similarly,  by \lemref{core}, there is a point $p_2\in \ch{A_2}$ such that
$p_2\in T_B$, and a point $p_5\in \ch{A_5}$ such that $p_5\in T_B$. Thus the segment $\overline{p_2p_5}$ is in $T_A$.

By \thmref{arrang} and our ordering of $A_1, \dots, A_6$, $\overline{p_1p_4}$ and $\overline{p_2p_5}$
intersect in some point $p$. That implies that $p\in T_A$ and $p\in
T_B$.  That provides the desired contradiction, since $T_A$ and $T_B$
are disjoint. 
\end{proof}

As a sharp contrast to \thmref{gd-main}, the following theorem shows that the situation is starkly different for unlabelled planar graphs. Namely, every set of $n$ lines supports all $n$-vertex unlabelled planar graphs. The proof of this theorem does not use any of the tools we introduced in the previous section and is in that sense elementary. It is not difficult to verify that the theorem also follows from the main result in \cite{pt04} which states the following: given a drawing of a graph $G$ in the plane where edges of $G$ are $x$-monotone curves any pair of which cross even number of times, $G$ can be redrawn as a straight-line crossing-free drawing where the $x$-coordinates of the vertices remain unchanged.

\begin{theorem}\thmlabel{gd-unlabelled}
\cite{pt04}
Given a set $L$ of $n$ lines in the plane, every planar graph has a straight line crossing free drawing where each vertex of $G$ is placed on a distinct line of $L$.
\textup{(}In other words, given any set $L$ of lines, labelled from $1$ to $n$, and any $n$-vertex planar graph $G$ there is a vertex labelling $\pi$ of $G$ such that $L$ supports $G$ with vertex labelling $\pi$.\textup{)}
\end{theorem}

\begin{proof}
In this proof we will use canonical orderings introduced in
\cite{dFPP90} and a related structure called frame introduced in \cite{untangling}. We first recall these tools. 
We can assume $G$ is an embedded edge maximal planar graph.\footnote{A planar graph $H$ is edge-maximal (also called, a \emph{triangulation}), if for all $vw\not\in E(H)$, the graph resulting from adding $vw$ to $H$ is not planar.} 
 Each face of \eg\ is bounded by a $3$-cycle. 
De Fraysseix \cite{dFPP90} proved that \eg\ has a vertex ordering $\sigma=(v_1,v_2, v_3, \dots,
v_n)$, called a \emph{canonical ordering}, with the following
properties. Define $G_i$ to be the embedded subgraph of \eg\ induced by
$\{v_1, v_2,\dots,v_i\}$.  Let $C_i$ be the subgraph of \eg\ induced
by the edges on the boundary of the outer face of $G_i$. Then
\begin{fortuna}
\item \xx, \yy\ and \zz\ are the vertices on the outer face of \eg. 
\item  For each $i\in\{3,4,\dots,n\}$, $C_i$ is a cycle containing $\xx\yy$.
\item  For each $i\in\{3,4,\dots,n\}$, $G_i$ is biconnected and \emph{internally $3$-connected}; that is, removing any two interior vertices of $G_i$ does not disconnect it.
\item For each $i\in\{3,4,\dots,n\}$, $v_i$ is a vertex of $C_i$ with at least two neighbours in $C_{i-1}$, and these neighbours are consecutive on $C_{i-1}$.
\end{fortuna}
For example, the ordering in \figref{canonical}(a) is a canonical ordering of the depicted embedded graph \eg.
\Figure{canonical}{
\includegraphics[width=3in]{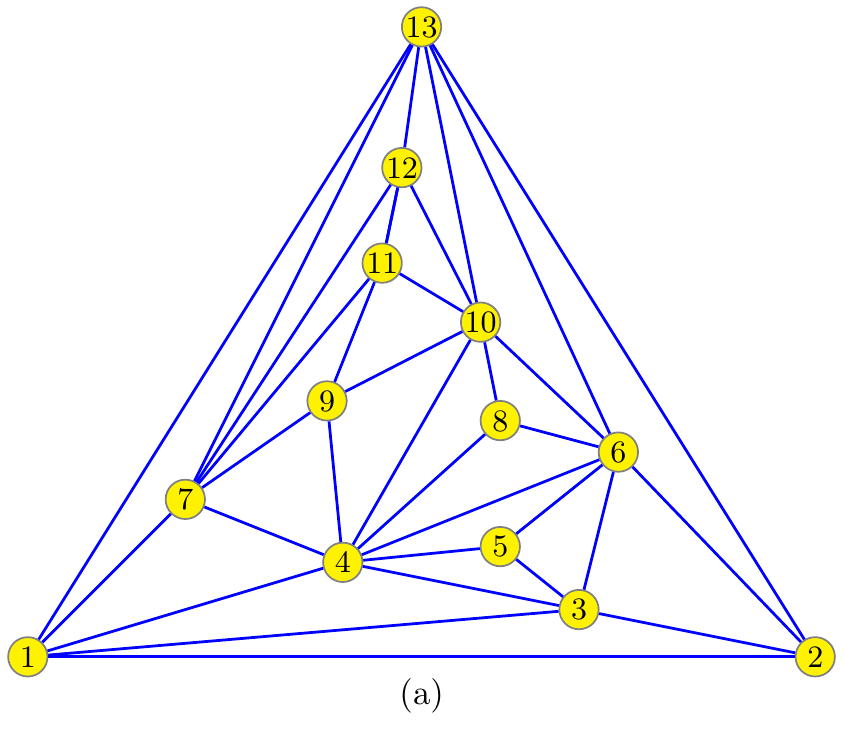}
\includegraphics[width=3in]{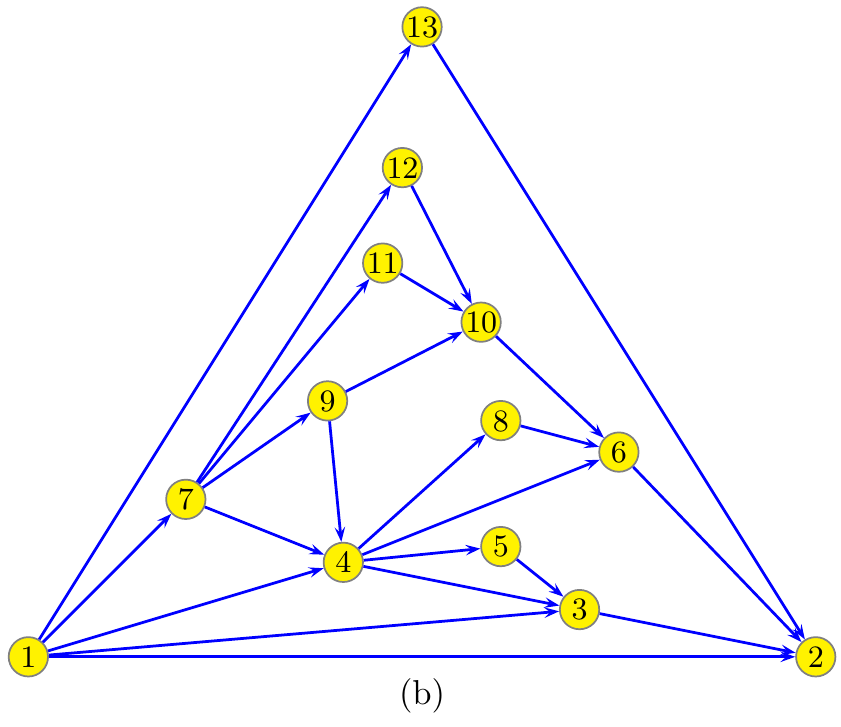}
}{Illustration for the proof of \thmref{gd-unlabelled}: (a) Canonical ordering of \eg, (b) Frame \fg\ of \eg}

A \emph{frame} \fg\ of \eg\ \cite{untangling} is the oriented subgraph of \eg\ with vertex set  $V(\fg):=V(\eg)$, where: 
\begin{fortuna}
\item \xx\yy\ is in $E(\fg)$ and is oriented from \xx\ to \yy.
\item For each $i\in\{3,4,\dots,n\}$ in the canonical ordering
  $\sigma$ of \eg, edges $pv_i$ and $v_ip'$ are in $E(\fg)$, where $p$
  and $p'$ are the first and the last neighbour, respectively, of
  $v_i$ along the path in $C_{i-1}$ from \xx\ to \yy\ not containing
  edge \xx\yy. Edge $pv_i$ is oriented from $p$ to $v_i$, and edge
  $v_ip'$ is oriented from $v_i$ to $p'$, as illustrated in
  \figref{canonical}(b). 
\end{fortuna}

By definition, \fg\ is a directed acyclic graph with one source \xx\,
and one sink \yy. The frame \fg\ defines a partial order \po\ on $V(\fg)$, where $v\po w$ whenever there is a directed path from $v$ to $w$ in \fg. 

Translate the given set $L$ of lines so that all vertices of the
arrangement of lines have negative $y$ coordinates, and
sort the lines $\ell_i\in L$ according to the $x$ coordinate $b_i$ of
the intersection of $\ell_i$ with the $x$ axis. 
Therefore, the lines $\ell_i\in L$ have equation $y=a_i(x-b_i)$,
with $b_1< b_2< \ldots < b_n$. Because all intersections
among lines of $L$ have negative coordinates, all $b_i$ are
distinct, and the values $1/a_i$ are sorted. 
Note that the slopes $a_i$ might be positive or negative.
Let $\hat{a} = \min|a_i|$. For any segment of slope in
$[-\hat{a},\hat{a}]$ connecting two points $(x_i,y_i)\in\ell_i$ and
$(x_j,y_j)\in\ell_j$ above the $x$ axis (that is, $y_i, y_j > 0$), 
$x_i<x_j$ if and only if $i<j$.

Construct a linear extension $v_{\rho(1)},v_{\rho(2)},\ldots,v_{\rho(n)}$ of the partial order
\po\ and define the bijection $\pi: V\rightarrow [n]$ as
$\pi(v_{\rho(i)})=i$. That is, the vertices of $G$ will be placed on the lines in
such a way that the partial order \po\  is
compatible with the order determined by the values $b_i$ of the lines.

We prove by induction that for every value $\hat{y}$ and every 
$i\geq 2$, it is possible to draw $G_i$ such that $v_1$ and $v_2$ are
placed on points $(b_1,0)$, $(b_n,0)$, and the $y$ coordinates of all
other vertices are in the horizontal slab $(0,\hat{y}]$. The base case ($i=2$) is
obviously true.

Note that we could have formulated the induction on the slopes of the
edges of $G_i$ in the drawing. In fact those two formulations imply
each other: for any value $0<s\leq\hat{a}$, there is a $\hat{y}_s>0$ such that any segment whose endpoints lie on distinct lines of $L$ and have $y$ coordinates
in $[0,\hat{y}_s]$, the slope of the segment is in $[-s,s]$. 
 This is
easy to see: draw an upward cone with apex on each point $(b_i,0)$ and bounded
by the lines of slopes $s$ and $-s$ through that point. Define
$\hat{y}_s$ as the $y$ coordinate of the lowest intersection point
between any two such cones. Any segment with a slope not in $[-s,s]$
and with its lowest point inside a cone must have its highest point
inside the same cone, therefore no segment connecting two different
lines inside the horizontal slab $[0,\hat{y}_s]$ can have such a slope.

Assume by induction that the statement is true for $G_{i-1}$. We will
show how to draw $G_i$ for a specific value $\hat{y}$.
The point $v_i$ will be placed on the point on line $\pi(v_i)$ with $y$
coordinate $\hat{y}$. 
Let $s_1$ and $s_2$ be the slopes of the segments $v_1v_i$ and $v_iv_2$, and
let $s=\max(|s_1|,|s_2|)/2$ or $\hat{a}$, whichever is smaller. 
Let $y_1$ be the intersection of the line
of slope $s$ through $v_i$ and line $\ell_1$ and $y_2$ the
intersection of the line of slope $-s$ through $v_i$ and
$\ell_n$. Note that $y_1$ and $y_2$ are strictly positive.
Let $\hat{y}' = \min(y_1,y_2,\hat{y}_{s})$. Apply the induction
hypothesis to draw $G_{i-1}$ in the horizontal slab $[0,\hat{y}']$.
Thus, in the drawing of $G_{i-1}$, all edges have slope at most
$s\leq\hat{a}$. Then by construction, the path in $C_{i-1}$ from \xx\
to \yy\ not containing edge \xx\yy\ is $x$-monotone (that is, all its
edges are oriented rightwards), and $v_i$ is above the supporting line of
each edge on that path. Therefore, $v_i$ can see all vertices in
$C_{i-1}$ and all edges adjacent to $v_i$ can be drawn.
\end{proof}
We conclude this part with an intriguing 3D variant of this graph
drawing problem. A graph is \emph{linkless} if it has an embedding in
3D such that any two cycles of the graph are unlinked\footnote{Two disjoint
curves in 3D are \emph{unlinked} if there is a continuous motion of
the curves which transforms them into disjoint coplanar circles
without one curve passing through the other or through itself.}.
These graphs form a three-dimensional analogue of the planar graphs.
\begin{open}
 Is there an arrangement of labelled planes in 3D such that any
 labelled linkless graph has a linkless straight-line embedding where
 each vertex is placed on the plane with the same label?
\end{open}

\section*{Acknowledgements}
We wish to thank J{\'a}nos Pach for bringing our attention to the notion of mutually avoiding sets \cite{mut, val} and for pointing out that \thmref{gd-unlabelled} was an easy consequence of \cite{pt04}. We also wish to thank Pat Morin for helpful discussions and the anonymous referees for their comments.



\begin{thebibliography}{GMPP91}

\bibitem[BDH{\etalchar{+}}09]{untangling}
Prosenjit Bose, Vida Dujmovi\'c, Ferran Hurtado, Stefan Langerman, Pat Morin,
  and David Wood.
\newblock A polynomial bound for untangling geometric planar graphs.
\newblock {\em Discrete and Computational Geometry}, 42:570--585, 2009.


\bibitem[Bou07]{bourbaki--a13}
Nicolas Bourbaki.
\newblock {\em \'El\'ements de math\'ematique. Alg\`ebre. Chapitres 1 \`a 3.}
\newblock Springer-Verlag Berlin Heidelberg, 2007.

\bibitem[CF90]{cf-dvrsi-90}
Bernard Chazelle and Joel Friedman.
\newblock A deterministic view of random sampling and its use in geometry.
\newblock {\em Combinatorica}, 10(3):229--249, 1990.

\bibitem[CSSS89]{csss-otass-89}
Richard Cole, Jeffrey Salowe, William Steiger, and Endre Szemer{\'e}di.
\newblock An optimal-time algorithm for slope selection.
\newblock {\em SIAM Journal on Computing}, 18(4):792--810, 1989.

\bibitem[DEK{\etalchar{+}}10]{mi-GD10}
Vida Dujmovi\'c, Will Evans, Stephen~G. Kobourov, Giuseppe Liotta, Christophe
  Weibel, and Stephen Wismath.
\newblock On graphs supported by line sets.
\newblock In {\em Proceedings of the 18th Symposium on Graph Drawing
  (GD'10)}, pages 177--182, Volume 6502, 2011.

\bibitem[dFPP88]{fpp-sssfe-88}
Hubert de~Fraysseix, J{\'a}nos Pach, and Richard Pollack.
\newblock Small sets supporting {Fary} embeddings of planar graphs.
\newblock In {\em Proceedings 20th Symposium on Theory of Computing (STOC)}, pages 426--433, 1988.

\bibitem[dFPP90]{dFPP90}
Hubert de~Fraysseix, J{\'a}nos Pach, and Richard Pollack.
\newblock How to draw a planar graph on a grid.
\newblock {\em Combinatorica}, 10(1):41--51, 1990.

\bibitem[Dol92]{Dolnikov-center-transversal}
Vladimir~L. Dol'nikov.
\newblock A generalization of the ham sandwich theorem.
\newblock {\em Mathematical Notes}, 52:771--779, 1992.

\bibitem[EBFK09]{1531118}
Alejandro Estrella-Balderrama, J.~Joseph Fowler, and Stephen~G. Kobourov.
\newblock Characterization of unlabeled level planar trees.
\newblock {\em  CGTA: Computational Geometry: Theory and Applications}, 42(6-7):704--721, 2009.

\bibitem[ES35]{es-cpg-35}
Paul Erd{\H o}s and George Szekeres.
\newblock A combinatorial problem in geometry.
\newblock {\em Compositio Mathematica}, 2:463--470, 1935.

\bibitem[GMPP91]{Gritzmann91}
Peter Gritzmann, Bojan Mohar, J{\'a}nos Pach, and Richard Pollack.
\newblock Embedding a planar triangulation with vertices at specified points.
\newblock {\em The American Mathematical Monthly}, 98(2):165--166, February 1991.

\bibitem[LS03]{ls02-opt-in-arr}
Stefan Langerman and William Steiger.
\newblock Optimization in arrangements.
\newblock In {\em Proceedings of the 20th International Symposium on
  Theoretical Aspects of Computer Science (STACS 2003)}, volume 2607 of {\em
  LNCS}, pages 50--61. Springer-Verlag, 2003.

\bibitem[Mat90]{m-cen-90}
Ji{\v r}i Matou{\v s}ek.
\newblock Construction of $\epsilon$-nets.
\newblock {\em Discrete and Computational Geometry}, 5:427--448, 1990.

\bibitem[Mat02]{Matousek-2002-LDG}
Ji{\v r}i Matou{\v s}ek.
\newblock {\em Lectures on Discrete Geometry}.
\newblock Springer-Verlag New York, Inc., Secaucus, NJ, USA, 2002.

\bibitem[PW01]{pw-epg-01}
J{\'a}nos Pach and Rephael Wenger.
\newblock Embedding planar graphs at fixed vertex locations.
\newblock {\em Graphs and Combinatorics}, 17:717--728, 2001.

\bibitem[RT86]{rt-rplbo-86}
Pierre Rosenstiehl and Robert~E. Tarjan.
\newblock Rectilinear planar layouts and bipolar orientations of planar graphs.
\newblock {\em Discrete Computational Geometry}, 1(4):343--353, 1986.

\bibitem[{\v Z}V90]{zv-ehst-90}
Rade~T. {\v Z}ivaljevi{\'c} and Sini{\v s}a~T. Vre{\'c}ica.
\newblock An extension of the ham sandwich theorem.
\newblock {\em Bulletin of the London Mathematical Society}, 22:183--186, 1990.

\bibitem[AEG{\etalchar{+}}94]{mut}
Boris Aronov, Paul Erd{\H o}s, Wayne Goddard, Daniel J. Kleitman, Michael Klugerman,  J{\'a}nos Pach, Leonard J. Schulman.
\newblock Crossing families.
\newblock {\em Combinatorica}, 14:127--134, 1994.


\bibitem[Val97]{val}
Pavel Valtr.
\newblock On mutually avoiding sets
\newblock {\em The Mathematics of Paul Erd{\H o}s II, Algorithms and Combinatorics}, 14:324--332, 1997.

\bibitem[AP95]{jp}
 Pankaj K. Agarwal and J{\'a}nos Pach.
\newblock {\em Combinatorial Geometry}.
\newblock Wiley-Interscience Series in Discrete Mathematics and Optimization, 1995.

\bibitem[PT04]{pt04}
J{\'a}nos Pach and Geza Toth.
\newblock Monotone drawings of planar graphs.
\newblock {\em Journal of Graph Theory}, 46(1):39--47, 2004.

\end{thebibliography}
\newcommand{\etalchar}[1]{$^{#1}$}

\end{document}